\documentclass[journal]{IEEEtran}
\usepackage{amssymb,stmaryrd,amsmath,amsfonts,rotating}
\usepackage[noadjust]{cite} \usepackage{color} 
\usepackage[vflt]{floatflt} \usepackage{epic}
\usepackage{color}

\newcommand{\x}{\underline{x}}

\newcommand{\xstab}{x_{\text{s}}}

\newcommand{\JIT}{\ensuremath{\text{\tiny JIT}}} 
\newcommand{\DEC}{\ensuremath{\text{\tiny DEC}}} 
\RequirePackage{bbm} 
\definecolor{TODO}{rgb}{0.6,0.6,0.6} 

\definecolor{TOCHECK}{rgb}{0.8,0.8,0.8} 

\newtheorem{theorem}{Theorem}
\newcommand{\btheo}{\begin{theorem}}
\newcommand{\etheo}{\end{theorem}}
\newcommand{\bproof}{\begin{proof}}
\newcommand{\eproof}{\end{proof}}
\newtheorem{definition}[theorem]{Definition}
\newcommand{\bdefi}{\begin{definition}}
\newcommand{\edefi}{\end{definition}}
\newtheorem{fact}[theorem]{Fact}
\newcommand{\bprop}{\begin{fact}}
\newcommand{\eprop}{\end{fact}}
\newtheorem{corollary}[theorem]{Corollary}
\newcommand{\bcor}{\begin{corollary}}
\newcommand{\ecor}{\end{corollary}}
\newtheorem{example}[theorem]{Example}
\newcommand{\bex}{\begin{example}}
\newcommand{\eex}{\end{example}}
\newtheorem{lemma}[theorem]{Lemma}
\newcommand{\blemma}{\begin{lemma}}
\newcommand{\elemma}{\end{lemma}}
\newtheorem{remark}[theorem]{Remark}
\newcommand{\bremark}{\begin{remark}}
\newcommand{\eremark}{\end{remark}}
\newtheorem{conj}[theorem]{Conjecture}
\newcommand{\bconj}{\begin{conj}}
\newcommand{\econj}{\end{conj}}

\newcommand{\naturals}{\ensuremath{\mathbb{N}}}
\newcommand{\integers}{\ensuremath{\mathbb{Z}}}

\def\0{{\tt 0}} 
\def\1{{\tt 1}} 
\def\?{{\tt *}} 
\newcommand{\MAPsmall}{\ensuremath{\text{\tiny MAP}}} 
\newcommand{\BPsmall}{\ensuremath{\text{\tiny BP}}} 
\newcommand{\qed}{{\hfill \footnotesize $\blacksquare$}}

\newcommand{\BEC}{\ensuremath{\text{BEC}}}

\newcommand{\dr}{d_\mathrm{r}}
\newcommand{\dl}{d_\mathrm{l}}

\allowdisplaybreaks
\begin{document} 
\title{Threshold Saturation on Channels with Memory via Spatial Coupling }
\author{\authorblockN{ Shrinivas Kudekar\authorrefmark{1} and Kenta Kasai\authorrefmark{2} \\ }
\authorblockA{\authorrefmark{1} New Mexico Consortium  and Center for Non-linear Studies, Los Alamos National Laboratory, NM, USA\\
 Email: skudekar@lanl.gov} \\
\authorblockA{\authorrefmark{2} Dept. of Communications and Integrated Systems, Tokyo Institute of Technology, 152-8550 Tokyo, Japan.\\
Email: {kenta}@comm.ss.titech.ac.jp} \\
 }

\maketitle
\begin{abstract}
We consider spatially coupled code ensembles. A particular instance
are convolutional LDPC ensembles.  It was recently shown that, for
transmission over the memoryless binary erasure channel, this coupling increases
the belief propagation threshold of the ensemble to the maximum
a-posteriori threshold of the underlying component ensemble. This paved the way for
a new class of capacity achieving low-density parity check codes. It was also
shown empirically that the same threshold saturation occurs when we consider
transmission over general binary input memoryless channels.   

In this work, we report on empirical evidence which suggests that the same phenomenon also
occurs when transmission takes place over a class of channels with memory. 
This is confirmed both by simulations as well as by computing EXIT
curves.
\end{abstract}

\section{Introduction} It has long been known that convolutional LDPC (or
spatially coupled) ensembles, introduced by Felstr{\"{o}}m and Zigangirov
\cite{FeZ99}, have excellent thresholds when transmitting over general
binary-input memoryless symmetric-output (BMS) channels. The fundamental reason
underlying this good performance was recently discussed in detail in
\cite{KRU10} for the case when transmission takes place over the binary erasure
channel (BEC).

In particular, it was shown in \cite{KRU10} that the BP threshold of the
spatially coupled ensemble (see the last paragraph of this section for a
definition) is essentially equal to the MAP threshold of the underlying
component ensemble. It was also shown that for long chains the MAP performance
of the chain cannot be substantially larger than the MAP threshold of the
component ensemble.  In this sense, the BP threshold of the chain is increased
to its maximal possible value.  This is the reason why they call this phenomena
{\em threshold saturation via spatial coupling}. In a recent paper
\cite{LeF10}, Lentmaier and Fettweis independently formulated the same
statement as conjecture. They attribute the observation of the equality of the
two thresholds to G. Liva.  The phenomena of threshold saturation seems not to
be restricted to the BEC.  It was also shown recently in \cite{KMRU10} that the
same phenomena manifests itself when we consider transmission over more general
BMS channels. 
  
The principle which underlies the good performance of spatially coupled
ensembles is very broad. It has been shown to apply to many other problems in
communications, and more generally computer science. To mention just a few, the
threshold saturation effect (dynamical threshold of the system being equal to
the static or condensation threshold) of coupled graphical models has 
recently been shown to occur for compressed sensing
\cite{KP10}, and a variety of graphical models in statistical physics and
computer science like the so-called $K$-SAT problem, random graph coloring, or
the Curie-Weiss model \cite{HMU10}. Other communication scenarios where the
spatially coupled codes have found immediate application is to achieve the
whole rate-equivocation region of the BEC wiretap channel \cite{RUAS10}. 

It is tempting to conjecture that the same phenomenon occurs for
transmission over general channels with memory. We provide some empirical
evidence that this is indeed the case. In particular, we compute 
 EXIT curves for transmission over a class of channels with memory known as the Dicode Erasure Channel (DEC). 
 We show that these curves behave in an identical fashion to the ones when 
 transmission takes place over the memoryless BEC. We also
compute fixed points (FPs) of the spatial configuration and we demonstrate
again empirically that these FPs have properties
identical to the ones in the BEC case.

For a review on the literature on convolutional LDPC ensembles we refer the
reader to \cite{KRU10} and the references therein.  As discussed in
\cite{KRU10}, there are many basic variants of coupled ensembles.  For the sake
of convenience of the reader, we quickly review the ensemble $(\dl, \dr, L,
w)$. This is the ensemble we use throughout the paper as it is the simplest to
analyze.

\subsection{$(\dl, \dr, L, w)$ Ensemble \cite{KRU10}}

We assume that the variable nodes are at sections $[-L, L]$, $L
\in \naturals$. At each section there are $M$ variable nodes, $M
\in \naturals$. Conceptually we think of the check nodes to be
located at all integer positions from $[- \infty, \infty]$.  Only
some of these positions actually interact with the variable nodes.
At each position there are $\frac{\dl}{\dr} M$ check nodes. It
remains to describe how the connections are chosen.  We assume that
each of the $\dl$ connections of a variable node at position $i$
is uniformly and independently chosen from the range $[i, \dots,
i+w-1]$, where $w$ is a ``smoothing'' parameter. In the same way,
we assume that each of the $\dr$ connections of a check node at
position $i$ is independently chosen from the range $[i-w+1, \dots,
i]$.

A discussion on the above ensemble and a proof of the following lemma can be found in \cite{KRU10}.
\begin{lemma}[Design Rate]\label{lem:designrate}
The design rate of the ensemble $(\dl, \dr, L, w)$, with $w \leq 2 L$,
is given by
\begin{align*}
R(\dl, \dr, L, w) & = 
(1-\frac{\dl}{\dr}) - \frac{\dl}{\dr} \frac{w+1-2\sum_{i=0}^{w} 
\bigl(\frac{i}{w}\bigr)^{\dr}}{2 L+1}.
\end{align*}
\end{lemma}

In the next section we provide the channel model and the joint iterative
decoder. We also present the density evolution analysis of the joint iterative
decoder when we consider $(\dl, \dr)$-regular LDPC ensembles. In the section on main results, 
we demonstrate the threshold saturation phenomena by using spatially coupled
codes. 

\section{Channels with Memory: The Dicode Erasure Channel}

The particular class of channel with memory that we consider is the Dicode
Erasure Channel (DEC). The DEC is a binary-input channel defined as follows. The
output of a binary-input linear filter $(1-D)$ ($D$ is the delay element) is
erased with probability $\epsilon$ and transmitted perfectly with probability
$1-\epsilon$. For this channel we will be interested in the symmetric
information rate (SIR), i.e., the capacity assuming i.i.d Bern(1/2) signalling. 
 In this case, the Shannon threshold for a given rate $r$ is given by
 $\frac{1-r}4 + \frac14\sqrt{(1-r)^2 + 8(1-r)}$.
The details on the definition of the channel and the analytical formula for the SIR 
 can be found in the thesis of Pfister \cite{Pfi03} and in \cite{PfS08}. 

\subsection{Joint Iterative Decoder, Density Evolution and the Extended BP Fixed Points} 

We use the joint iterative decoder (JIT) of Pfister and Siegel \cite{PfS08}. More
precisely, we consider a turbo equalization system, which performs one channel
iteration (BCJR step) for each iteration over the LDPC code. As a result, in
every iteration, first the channel detector uses the extrinsic information
provided by the LDPC code to compute its extrinsic erasure fraction. This is
then fed to the LDPC decoder which then again computes the usual variable node
and check node erasure messages. 

The simplicity of the DEC gives an analytical formula for the erasure fraction
of the message which is passed from the channel detector to the LDPC code (see
\cite{PfS08} for a derivation). This is given by $$f(x) = \frac{4\epsilon^2}{(2
- x(1-\epsilon))^2},$$ where $x$ represents the fraction of erasures entering
the channel detector from the LDPC code. $f(.)$ represents the extrinsic
erasure information provided by the channel detector. 

To summarize: the density evolution\footnote{See \cite{PfS08} for a rigorous
justification of the density evolution analysis.} (DE) equation for the
 case of $(\dl, \dr)$-regular LDPC ensemble is given  by

\begin{align*}
x = f((1-(1-x)^{\dr-1})^{\dl}) (1 - (1 - x)^{\dr-1})^{\dl-1}.
\end{align*}
Note that the term inside the brackets in $f(.)$ represents the probability that
a variable node is in erasure as given by the LDPC code. Also it is not hard to see that $f(x)\leq 1$ for any $x$.   

\begin{example} Consider JIT decoding of the DEC with
$(5,15)$-regular LDPC ensemble. The design rate of this code is $2/3$. 
 Using the SIR formula ($=1 - 2\epsilon^2/(1+\epsilon)$) from \cite{PfS08} we
 get that the Shannon threshold at rate=2/3 is given by
 $\epsilon^{\text{Sh}}_{\DEC}=0.5$. Figure~\ref{fig:5-15-BP} shows the performance of
 the JIT decoder. We see that the threshold is given by
 $\epsilon^{\JIT}_{\DEC}(5,15)\approx 0.363471$, which is far away from the capacity.
 Throughout the paper we will use $\epsilon^{\JIT}_{\DEC}(\dl, \dr)$ to denote the
 threshold of the JIT decoder when we use $(\dl, \dr)$-regular LDPC
 ensemble and transmit over the DEC. 
\begin{figure}[htp] \centering
\input{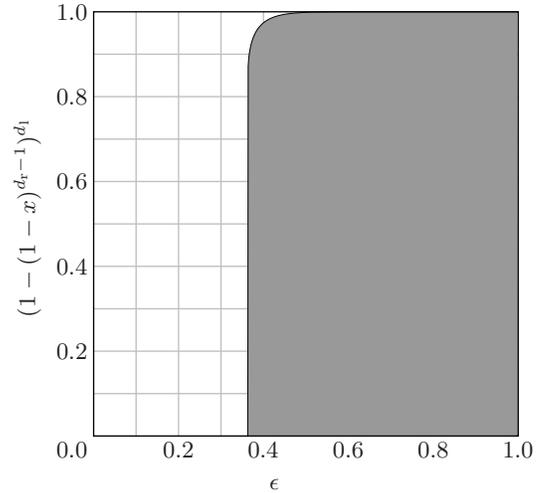} 
\caption{\label{fig:5-15-BP} 
The BP curve for the $(\dl=5,
\dr=15)$-regular ensemble and transmission over the DEC. The threshold of the
JIT decoder is given by $\epsilon^{\JIT}_{\DEC}(5,15)\approx 0.363471$.} 
\end{figure}
\end{example}

\subsubsection*{The EXIT curve} The EXIT curve\footnote{To be very precise, we should call the curves we plot as EXIT-like curves. The reason being that we do not provide any operation interpretation of these curves, like the Area theorem \cite{RiU08} in this work. The curves serve only to illustrate the capacity achieving nature of coupled-codes.} plots all the fixed-points of
the DE equation. The curve is given by the parametric curve $\{
(1-(1-x)^{\dr-1})^{\dl}, \epsilon(x)\}$. We obtain $\epsilon(x)$ by solving for
$\epsilon$ in the DE equation. 

As an example, we plot the EXIT curve for various $(\dl, \dr)$-regular LDPC
ensembles as shown in Figure~\ref{fig:5-15-EBP}. The JIT
threshold is got by dropping a vertical line from the leftmost point on any
given curve. We note that for every $\epsilon > \epsilon^{\JIT}_{\DEC}(\dl, \dr)$,
there are exactly 3 fixed-points. One of them being the trivial 0 fixed-point.
This ``C'' shape of the EXIT curve is also what we observe when we transmit
through a memoryless BEC using $(\dl, \dr)$-regular LDPC ensemble. Also we
remark that as the degrees increase, keeping the design rate fixed, the
 JIT threshold keeps on decreasing. This is also the case for
transmission over memoryless BEC. In fact, for memoryless BEC case, the BP
threshold goes to zero as we increase the degrees. 
\begin{figure}[htp] \centering
\input{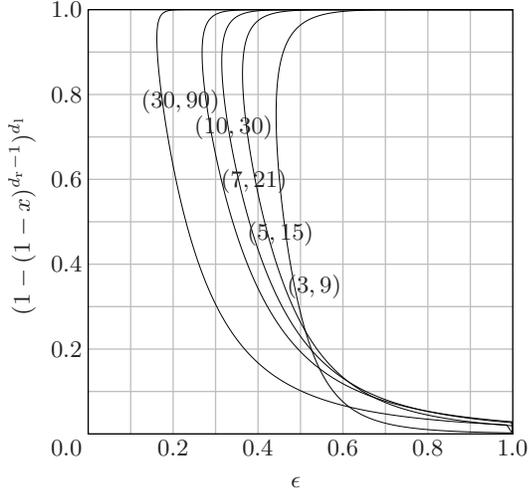} 
\caption{\label{fig:5-15-EBP} 
The EXIT curve for regular LDPC ensembles with $(\dl, \dr)$ given by  $(3,9)$, $(5,15)$, $(7,21)$, $(10,30)$, $(30,90)$, and transmission over the
DEC. We observe that the JIT threshold moves to the left and
eventually will go to zero as degrees go to infinity.} 
\end{figure}
 We can also show the same result for the DEC. More precisely, we have
 \begin{lemma}[JIT Threshold Goes to Zero]\label{lem:BPgoesto0}
 For any $(\dl, \dr)$-regular ensemble we have
 \begin{align*}
 \epsilon^{\JIT}_{\DEC}(\dl, \dr) \leq \sqrt{\frac{1}{\sqrt{\dr-1}(1 - (\dl - 1)e^{-\sqrt{\dr-1}})}}.
 \end{align*}
 \end{lemma}
\begin{proof}
We claim that the
necessary condition for the JIT decoder to succeed is given by 
$$
\epsilon^2 (1 - (1-x)^{\dr-1})^{\dl-1} < x,
$$
for all $x\in (0,1]$. Indeed, suppose on the contrary that there exists a $c\in
(0,1]$ such that the above inequality is violated. Thus we have 
$
\epsilon^2 (1 - (1-c)^{\dr-1})^{\dl-1} \geq c.
$
Since $f(x)\geq \epsilon^2$ for all $x\in [0,1]$ we get
$$
f(c) (1 - (1-c)^{\dr-1})^{\dl-1} \geq c.
$$
This implies that there exists a FP of DE for the DEC for some value in $[c,1]$.
 It is not hard to see that this implies the JIT decoder will get stuck at this
 FP, resulting in unsuccessful decoding. 

Thus we must have that for all $x\in (0,1]$
$$
\epsilon^2 (1 - (1-x)^{\dr-1})^{\dl-1} < x. 
$$
For the choice of $x=\frac1{\sqrt{\dr-1}}$ we get the statement of the lemma.
 To see this computation first write $(1-x)^{\dr-1}$ as $e^{(\dr-1)\log(1-x)}$.
 Then use $\log(1-x) \leq -x$ and $x=\frac1{\sqrt{\dr-1}}$ to get 
$(1-x)^{\dr-1} \leq e^{-\sqrt{\dr-1}}$. After this use 
\begin{align*}
(1 - e^{-\sqrt{\dr-1}})^{\dl-1} & = 1 - (1 -(1 - e^{-\sqrt{\dr-1}})^{\dl-1}) \\ 
& \geq 1 - (\dl-1)e^{-\sqrt{\dr-1}},
\end{align*}
to complete the argument. 
\end{proof}
As a consequence of Lemma~\ref{lem:BPgoesto0} we get that,
with the ratio $\dl/\dr$ kept fixed, $\lim_{\dl \to
\infty}\epsilon^{\JIT}_{\DEC}(\dl, \dr) = 0$.

\section{Main Results} In this section we show, empirically, that
spatially coupled-codes achieve the Shannon capacity of the DEC.  We recall
that we are consider SIR which is give by the formula SIR$=1 -
2\epsilon^2/(1+\epsilon)$. For the sake of exposition, we demonstrate our
results only for rate equals $2/3$. The Shannon threshold for this rate is
given by $\epsilon^{\text{Sh}}_{\DEC}=0.5$. For other rates similar results can be
observed.  From the preceding section we see that standard $(\dl,
\dr)$-regular LDPC ensembles do not saturate the JIT threshold (to the
Shannon threshold).

We begin by writing down the DE equation for the coupled-codes. 

\subsection{Density Evolution}

Consider the $(\dl, \dr, L, w)$ ensemble. Recall that there are $2L+1$ sections
of variable nodes.  Each section has $M$ variable nodes. We transmit variable
nodes sectionwise over the DEC. More precisely, the variable nodes in section
$-L$ are transmitted first, followed by variable nodes in section $-L+1$ and so
on so forth till we finally transmit all the variable node in section $L$. As a
consequence we have a channel detected factor graph sitting on top of each
section of the coupled-code. 

To perform the DE analysis, we  already take the limit $M\to \infty$. As a
result of this limit, one can ignore the boundary effects of the channel
detector and treat the channel detectors as disconnected\footnote{Another way to
think about this is to imagine that we transmit a known sequence of bits of
length equal to the memory of the channel after we transmit all the variable nodes in
each section. Since the channel memory is finite, this induces a rate loss going
to zero as $M\to \infty$. Now the known sequence is the initial state for each
of the channel detectors and hence we can consider them disconnected.}. 

Let $x_i$, $i\in \integers$, denote the average erasure probability which
is emitted by variable nodes at position $i$. For $i \not \in [-L, L]$
we set $x_i=0$.
For $i \in [-L, L]$ the DE is given by
\begin{align}\label{eq:densevolxi}
x_i 
& = \epsilon_i \Bigl(1-\frac{1}{w} \sum_{j=0}^{w-1} \bigl(1-\frac{1}{w}
\sum_{k=0}^{w-1} x_{i+j-k} \bigr)^{\dr-1} \Bigr)^{\dl-1},
\end{align}
where $\epsilon_i$ is given by
\begin{align}\label{eq:channelEXIT}
\epsilon_i = f\Big(\Bigl(1-\frac{1}{w} \sum_{j=0}^{w-1} \bigl(1-\frac{1}{w}
\sum_{k=0}^{w-1} x_{i+j-k} \bigr)^{\dr-1} \Bigr)^{\dl}\Big),
\end{align}
where recall that $f(\cdot)$ is the channel extrinsic transfer function. 
 We will use the notation $\epsilon^{\JIT}_{\DEC}(\dl, \dr, L, w)$ to denote the
 threshold of the JIT decoder when we use the $(\dl, \dr, L, w)$
 ensemble for transmission. As a shorthand we use
 $
g(x_{i-w+1},\dots,x_{i+w-1})$ to denote $\Bigl(1-\frac{1}{w} \sum_{j=0}^{w-1} \bigl(1-\frac{1}{w}
\sum_{k=0}^{w-1} x_{i+j-k} \bigr)^{\dr-1} \Bigr)^{\dl-1}.$

\begin{definition}[FPs of Density Evolution]\label{def:fixedpoints}
Consider DE for the $(\dl, \dr, L, w)$ ensemble.
Let $\x=(x_{-L}, \dots, {x}_{L})$. We call $\x$ the {\em constellation}. We say that $\x$ forms a FP
of DE with channel $\epsilon$ if $\x$ fulfills (\ref{eq:densevolxi})
for $i \in [-L, L]$.  As a shorthand we then say that $(\epsilon, \x)$
is a FP.  We say that $(\epsilon, \x)$ is a {\em non-trivial}
FP if $\x$ is not identically equal to $0\,\forall\, i$.
Again, for $i\notin [-L,L]$, $x_i = 0$.
\qed
\end{definition}

\begin{definition}[Forward DE and Admissible Schedules]\label{def:forwardDE} 
Consider {\em forward} DE for the $(\dl, \dr, L, w)$ ensemble.  More
precisely, pick a channel $\epsilon$. Initialize 
$\x^{(0)}=(1, \dots, 1)$. Let $\x^{(\ell)}$ be the result of
$\ell$ rounds of DE. More precisely, $\x^{(\ell+1)}$ is generated from
$\x^{(\ell)}$ by applying the DE equation \eqref{eq:densevolxi} to each
section $i\in [-L,L]$,
\begin{align*}
x_i^{(\ell+1)} & = \epsilon_i g(x_{i-w+1}^{(\ell)},\dots,x_{i+w-1}^{(\ell)}).
\end{align*}
We call this the {\em parallel} schedule.
 The important difference with the memoryless BEC case is that the channel
 $\epsilon_i$ is not fixed for the DEC and decreases with increasing iterations according to
 \eqref{eq:channelEXIT}.

More generally, consider a schedule in which in each step $\ell$
an arbitrary subset of the sections is updated, constrained only by
the fact that every section is updated in infinitely many steps. We
call such a schedule {\em admissible}. Again, we call $\x^{(\ell)}$
the resulting sequence of constellations.  \qed 
\end{definition}
One can show that if we perform forward DE under any admissible schedule, then
the constellation $\x^{(\ell)}$ converges to a FP of DE and this FP is
independent of schedule. This statement can be proved similar to the one in
\cite{KRU10}.

\subsection{Forward DE -- Simulation Results}

We consider forward DE for the $(\dl, \dr, L, w)$ ensemble. More precisely, we
fix an $\epsilon$ and initialize all $x_i$ for $i \in [-L,L]$ to 1. Then we run
the DE given by \eqref{eq:densevolxi} till we reach a fixed-point. We fix
$L=250$. For $\dl=3$ and $\dr=9$, we have that $\epsilon^{\JIT}_{\DEC}(3, 9, 300, 3)
\approx 0.49815$. If we increase the degrees we get $\epsilon^{\JIT}_{\DEC}(5, 15,
300, 5) \approx 0.49995$, $\epsilon^{\JIT}_{\DEC}(7, 21, 300, 7) \approx 0.499989$ and
$\epsilon^{\JIT}_{\DEC}(9, 27, 300, 9) \approx 0.499996$. We observe that for
increasing the degrees the threshold approaches the Shannon threshold of $0.5$.

\subsection{The EXIT Curve for Coupled Ensembles}

We now come to the key point of the paper, the computation of the EXIT curve. 
 Before we do this, we define the entropy of a constellation $\x = (x_{-L},\dots,x_{L})$ as
 \begin{align*}
 \chi = \frac1{2L+1}\sum_{i=-L}^{L} x_i.
 \end{align*}
 
To plot the EXIT curve we first fix $\chi\in [0,1]$ and then run DE such
that the resulting FP constellation has entropy equal to $\chi$. This is the
reverse DE procedure as described in \cite{MMRU09}. We remark that $f(x)$ is an
increasing function of $\epsilon$, hence in the reverse DE procedure one can
easily find an appropriate $\epsilon$ by the bisection method. 

Figure~\ref{fig:EBPEXIT_coupled} shows the plot of the EXIT curve for the
$(5,15,L,5)$ ensemble with $L=2, 4, 8, 16, 32, 64, 128, 256, 512$. We see that the
curves look very similar to the curves when transmitting over a BMS channel. For
very small values of $L$, the curves are far to the right due to significant
rate loss that is incurred at the boundary. As $L$ increases the rate loss
diminishes and the JIT threshold is very close to the
Shannon threshold. This picture strongly suggests that the same threshold
saturation effect ($\epsilon^{\JIT}_{\DEC}(\dl, \dr, L, w) \approx \epsilon^{\MAPsmall}_{\DEC}(\dl,
\dr, L, w)$) also occurs for the DEC as it was shown analytically in
\cite{KRU10}. 

\begin{figure}[htp] \centering
\input{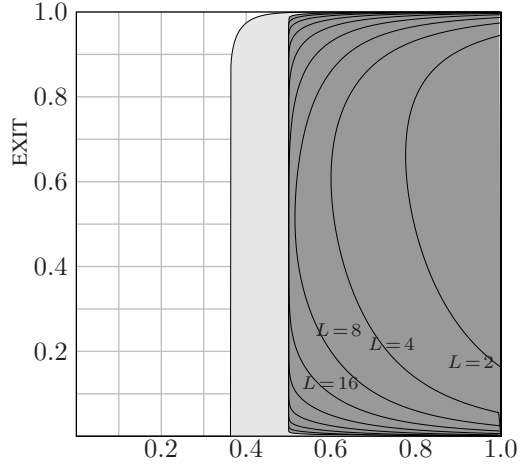} 
\caption{\label{fig:EBPEXIT_coupled} 
The EXIT curve for the $(\dl=5,
\dr=15, L, 5)$ ensemble and transmission over the DEC for $L=2,4,8,16,32,64,128,256,512$. 
The curves keep moving to the left as $L$ increases similar to the curves when
transmitting over BMS. The ``vertical'' drop in the EXIT curves occurs at
$\approx 0.5$ for $L\geq 32$. Also shown in light gray is the BP exit curve for
 the uncoupled $(5,15)$-regular ensemble.} 
\end{figure}

\subsection{Shape of Fixed Point of Density Evolution}\label{sec:propertiesofFP}
We plot the constellation representing the unstable FP of DE. This FP cannot be
reached via forward DE and is obtained via reverse DE procedure. We recall 
 that this FP played a key role in proving the threshold saturation phenomena
 when transmitting over the BEC. Let us describe the (empirically observed) crucial properties of this
 constellation. 
 \begin{itemize}
 \item[(i)] The constellation is symmetric around $i=0$ and is unimodal. The
 constellation has $\epsilon\approx 0.49995$. 
 \item[(ii)] Let $\xstab(\epsilon)$ denote a stable FP of DE. The value in the
 flat part in the middle is $\approx 0.4434$ which is very close to the stable
 FP of DE for the underlying uncoupled $(5,15)$-regular ensemble.  
 \item[(iii)] The transition from close to zero to close to $\xstab(\epsilon)$
 is very quick.
 \end{itemize}

\begin{figure}[htp] \centering
\input{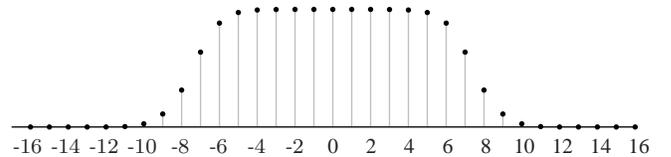} 
\caption{\label{fig:FPconstellation} 
The constellation representing FP of DE for $(5,15,33,5)$ ensemble and entropy
fixed to $\chi=0.2$. This is an unstable FP constellation. The
constellation is very similar to any unstable FP constellation when transmitting
over memoryless BEC. The constellation is unimodal. There is a long tail of zeros
followed by a sharp transition and then a long flat part with values close to
$\xstab(\epsilon)$. The constellation has $\epsilon\approx 0.49995$.} 
\end{figure}

\section{A Possible Proof Approach}
Till now we gave empirical evidence of the threshold saturation phenomena when
transmitting over the DEC using coupled-codes.   
Before we proceed to give the proof idea for the threshold saturation, we first show that coupling indeed helps. 
 More precisely we have the following lemma,
\begin{lemma}[Spatial Coupling Helps] \label{lem:BPofcoupled}
For $\dl, \dr \to \infty$ with the ratio $\dl/\dr$ kept fixed, we have
\begin{align*}
\epsilon^{\JIT}_{\DEC}(\dl, \dr, L, w) \geq \frac{\dl}{\dr}.
\end{align*}
\end{lemma}
\begin{proof}
Since $\epsilon_i$ is an increasing function of  $x_{i-w+1},\dots,x_{i+w-1}$, we have
 $\epsilon_i \leq f(1) \leq \frac{4\epsilon^2}{(1+\epsilon)^2} \leq \epsilon.$
Combining this with the DE equation for the coupled-codes, we get
\begin{align*}
x_i \leq \epsilon g(x_{i-w+1}, \dots, x_{i+w-1}),
\end{align*}
for all $i\in [-L,L]$. 
 But we know from Theorem 10 in \cite{KRU10} that
 $\lim_{\dl\to \infty}\epsilon^{\BPsmall}_{\BEC}(\dl, \dr, L, w)\to \frac{\dl}{\dr}$. Thus for
 $\epsilon<\frac{\dl}{\dr}$ the right-hand-side of the above inequality goes to
 zero. 
 Hence the lemma. 
\end{proof}

As an example, consider the $(\dl, \dr)$-regular ensemble with
$\dl/\dr=1/3$ (rate equal to $2/3$) . For $L\to \infty$, the rate of the $(\dl,
\dr, L, w)$ goes to $2/3$. From Lemma~\ref{lem:BPgoesto0} we have that
$\epsilon^{\JIT}_{\DEC}(\dl, \dr) \to 0$ and from Lemma~\ref{lem:BPofcoupled} 
 we have that $\epsilon^{\JIT}_{\DEC}(\dl, \dr, L, w) \geq
 \frac{\dl}{\dr}=\frac13$. Thus spatial coupling indeed boosts the JIT
 threshold. 
 However the empirical evidence suggests that the boost is all the
 way up to the Shannon threshold (which is $0.5$ in this case). Since there is ample similarity
 between the DEC and the BEC,
 the guideline for a proof is similar to when we are
 transmitting over the BEC. \\
{\em (i) Existence of FP:} A key ingredient in proving the result for the BEC
was to show the existence of a special FP of DE $(\x, \epsilon^*)$. In principle, the BEC proof
should extend. The only difference is that instead of a constant channel
$\epsilon$, we have a channel value which depends on the FP constellation
itself. However, since the functions involved are rational, this should not 
 be a big hurdle.\\
{\em (ii) Shape of the constellation and the transition length:} 
The next task is to show that the FP guaranteed by the above theorem has the
properties as given in Section~\ref{sec:propertiesofFP}. Proving this would
 first involve showing that the underlying regular ensemble has a ``C'' shaped 
 EXIT curve. Intuitively, this means that the FP constellation (of the
 coupled-code) can only hover around the stable FPs of DE (of the underlying
 regular ensemble), implying that it
 has either a large tail of zeros or a large flat part with values close to
 $\xstab(\epsilon^*)$.
\\
{\em (iii) Construction of the EXIT curve and the Area Theorem:} 
Another key part of the BEC proof was to construct a family of FPs (not
necessarily stable FPs) using the special FP guaranteed by the Existence theorem. 
The  EXIT curve plus the fast transition would allow us to show that this
special FP must have an associated channel parameter, $\epsilon^*$, very close to the Shannon
threshold (for large degrees.)\footnote{For finite degrees, $\epsilon^*$ should be very
close to the MAP threshold of the $(\dl, \dr)$-regular ensemble. One should be able to prove this by formulating an appropriate Area theorem (see Section 3.20 in \cite{RiU08}).} \\
{\em Operational interpretation:} The proof would be completed by providing an
operation meaning to the  EXIT curve. Loosely speaking, the EXIT
constructed above would have a vertical drop at $\epsilon \approx
\epsilon^{\text{Sh}}(\dl, \dr)$ (cf. Figure~\ref{fig:EBPEXIT_coupled}).
 This would help to show that for any $\epsilon<\epsilon^{\text{Sh}}(\dl, \dr)$, the 
 JIT decoder will go to the trivial FP. 

\section{Conclusions}
In this paper we show that empirically coupled-codes saturate the JIT 
 threshold on the DEC. For the channel extrinsic transfer function we consider the case
when there is no precoding. 
We list below some comments and open questions.
\begin{itemize}
\item[(i)]
An obvious future direction is to complete the proof of threshold saturation.
The guidelines provided above serve as a starting point. Following this route,
in principle, it should be possible to prove the capacity achieving nature of
these codes on the DEC.  
\item[(ii)]
Another interesting question is that whether the threshold saturation phenomena
can be shown to be true for all channel extrinsic transfer functions $f(.)$ which are non-decreasing both in
$\epsilon$ and $x$ (threshold saturation holds when $f(.)$ represents precoding).
\item[(iii)]
A proof of the threshold saturation phenomena should also pave the way for the
justification of the Maxwell construction to determine
$\epsilon^{\MAPsmall}_{\DEC}(\dl, \dr)$ for the DEC. 
\item[(iv)]
Recently, it was observed that coupled MacKay-Neal (MN) codes with bounded degree exhibit the BP threshold very close to the Shannon threshold over the BEC \cite{SCMN10}. 
It is interesting to see if the coupled MN codes have the JIT threshold close to the SIR over the DEC.
\end{itemize}

\section{Acknowledgments}
SK acknowledges support of NMC via the NSF collaborative grant CCF-0829945
on ``Harnessing Statistical Physics for Computing and Communications.'' SK
would also like to thank R\"udiger Urbanke, Misha Chertkov and Henry Pfister 
for their encouragement.

\bibliographystyle{IEEEtran} 
\bibliography{lth,lthpub,kasai}
\end{document}